%% file: ipcf-conf.tex
\documentclass[11pt]{entcs}

\usepackage{entcs-kavvos}
\usepackage{graphicx}
\usepackage[ipcf]{optional}
\usepackage{amsmath}
\usepackage{amssymb}
\usepackage{mathtools}
\usepackage{multicol}
\usepackage{stmaryrd}
\usepackage{prooftree}
\usepackage{enumitem}
\usepackage{framed}
\usepackage{tikz-cd}

\newcommand{\sem}[2]{\left\llbracket #1 \right\rrbracket^{#2}}
\newcommand{\defeq}{\stackrel{\mathclap{\mbox{\tiny def}}}{=}}

\newcommand{\ctxt}[2]{#1\mathbin{;}#2}
\newcommand{\ibox}[1]{\mathsf{box\;}#1}
\newcommand{\letbox}[3]{\mathsf{let\;box\;} #1 \Leftarrow #2 \mathsf{\;in\;} #3}
\newcommand{\fixlob}[2]{\mathsf{fix\;} #1 \mathsf{\;in\;} #2}

\newcommand{\red}{\mathrel{\longrightarrow}}
\newcommand{\redp}{\mathrel{\Longrightarrow}}
\newcommand{\redt}{\mathrel{\longrightarrow^\ast}}

\newcommand{\fv}[1]{\textsc{fv}\left(#1\right)}
\newcommand{\ufv}[1]{\textsc{fv}_0\left(#1\right)}
\newcommand{\bfv}[1]{\textsc{fv}_{\geq 1}\left(#1\right)}
\newcommand{\vars}[1]{\textsc{Vars}\left(#1\right)}

\newlist{indproof}{itemize}{5}
\setlist[indproof]{%
  itemsep=5pt,  % space between items
  font={\sc}, % set the label font
  label={}
}

\newcommand{\indcase}[1]{\item{\sc Case({\normalfont #1})}.}

\makeatletter
\newcommand{\dotdiv}{\mathbin{\text{\@dotdiv}}}

\newcommand{\@dotdiv}{%
  \ooalign{\hidewidth\raise1ex\hbox{.}\hidewidth\cr$\m@th-$\cr}%
}
\makeatother

\def\lastname{Kavvos}

\begin{document}

\begin{frontmatter}

\title{Intensionality, Intensional Recursion, \\
  and the G\"odel-L\"ob axiom}

\thanks[myemail]{This is a revised version of the third chapter of
  \cite{Kavvos2017}, which is in turn based on a paper presented at
  the 7th Workshop on Intuitionistic Modal Logic and Applications
  (IMLA 2017). Email: \href{mailto:g.a.kavvos@gmail.com}
  {\texttt{\normalshape g.a.kavvos@gmail.com}}. Current
  affiliation: Department of Computer Science, Aarhus University.}

\author{G. A. Kavvos}

\address{Department of Computer Science \\
  University of Oxford\\
  Oxford, United Kingdom}

\begin{abstract}
  The use of a necessity modality in a typed $\lambda$-calculus
  can be used to separate it into two regions. These can be
  thought of as intensional vs. extensional data: data in the
  first region, the modal one, are available as code, and their
  description can be examined. In contrast, data in the second
  region are only available as values up to ordinary equality.
  This allows us to add non-functional operations at modal types
  whilst maintaining consistency. In this setting, the
  G\"odel-L\"ob axiom acquires a novel constructive reading: it
  affords the programmer the possibility of a very strong kind of
  recursion which enables them to write programs that have access
  to their own code. This is a type of computational reflection
  that is strongly reminiscent of Kleene's Second Recursion
  Theorem.
\end{abstract}


\end{frontmatter}

%%% ARTICLE BEGINS HERE

\section{Introduction}

This paper is about putting a logical twist on two old
pieces of programming lore:

\begin{itemize}
  \item
    First, it is about using \emph{modal types} to treat
    \emph{programs-as-data} in a type-safe manner.

  \item
    Second, it is about noticing that---in the context of
    intensional programming---a constructive reading of the
    G\"odel-L\"ob axiom, i.e. $\Box(\Box A \rightarrow A)
    \rightarrow \Box A$, amounts to a strange kind of recursion,
    namely \emph{intensional recursion}.
\end{itemize}

\noindent We will introduce a \emph{typed $\lambda$-calculus with
modal types} that supports both of these features. We will call it
\emph{Intensional PCF}, after the simply-typed $\lambda$-calculus
with $\mathbf{Y}$ introduced by Scott \cite{Scott1993} and Plotkin
\cite{Plotkin1977}.

\subsection{Intensionality and Programs-as-data}

To begin, we want to discuss our notion of
\emph{programs-as-data}. We mean it in a way that is considerably
stronger than the higher-order functional programming with which
we are already familiar, i.e. `functions as first-class
citizens.' In addition to that, our notion hints at a kind of
\emph{homoiconicity}, similar to the one present in the
\textsc{Lisp} family of languages. It refers to the ability given
to a programmer to \emph{quote} code, and carry it around as a
datum; see \cite{Bawden1999} for an instance of that in
\textsc{Lisp}. This ability can be used for
\emph{metaprogramming}, which is the activity of writing programs
that write other programs. Indeed, this is what \textsc{Lisp}
macros excel at \cite{Graham1993}, and what the
\emph{metaprogramming community} has been studying for a long
time; see e.g. \cite{Taha2000,Tsukada2010}. Considering programs as
data---but in an untyped manner---was also the central idea in the
work of the \emph{partial evaluation community}: see
\cite{Jones1993,Jones1996,Jones1997}.

But we would like to go even further. In \textsc{Lisp}, a program
is able to process code by treating it as mere symbols, thereby
disregarding its function and behaviour. This is what we call
\emph{intensionality}: an operation is \emph{intensional} if it is
\emph{finer than equality}. This amounts to a kind of
\emph{non-functional computation}. That this may be done
type-theoretically was suspected by Davies and Pfenning
\cite{Davies2001,Davies2001a}, who introduced modal types to
programming language theory. A system based on nominal techniques
that fleshed out those ideas was presented by Nanevski
\cite{Nanevski2002}. The notions of intensional and extensional
equality implicit in this system were studied using logical
relations by Pfenning and Nanevski \cite{Nanevski2005}. However,
none of these papers studied whether the induced equational
systems are consistent. We show that, no matter the intensional
mechanism at use, modalities enable consistent intensional
programming.

To our knowledge, this paper presents the first consistency proof
for intensional programming.

\subsection{Intensional Recursion}

We also want to briefly explain what we mean by \emph{intensional
recursion}; a fuller discussion may be found in
\cite{Abramsky2014,Kavvos2017}. Most modern programming languages
support \emph{extensional recursion}: in the body of a function
definition, the programmer may make a finite number of calls to
the definiendum itself. Operationally, this leads a function to
examine its own values at a finite set of points at which it has
hopefully already been defined. In the \emph{untyped
$\lambda$-calculus}, with $=_\beta$ standing for
$\beta$-convertibility, this is modelled by the \emph{First
Recursion Theorem (FRT)} \cite[\S 6.1]{Barendregt1984}:

\begin{thm}
  [First Recursion Theorem]
  $\forall f \in \Lambda.\
  \exists u \in \Lambda.\
  u =_\beta fu$.
\end{thm}

However, as Abramsky \cite{Abramsky2014} notes, in the
\emph{intensional paradigm} we have described above a stronger
kind of recursion is attainable. Instead of merely examining the
result of a finite number of recursive calls, the definiendum can
recursively have access to a \emph{full copy of its own source
code}. This is embodied in Kleene's \emph{Second Recursion Theorem
(SRT)} \cite{Kleene1938}. Here is a version of the SRT in the
untyped $\lambda$-calculus, where $\ulcorner u \urcorner$ means
`the G\"odel number of the term $u$' \cite[\S 6.5, Thm.
6.5.9]{Barendregt1984}.

\begin{thm}
  [Second Recursion Theorem]
  $\forall f \in \Lambda.\
  \exists u \in \Lambda.\
    u =_\beta f \, \ulcorner u \urcorner$.
\end{thm}

\noindent Kleene also proved the following, where $\Lambda^0$ is
the set of closed $\lambda$-terms:

\begin{thm}
  [Existence of Interpreter]
  \label{thm:lambdainterp}
  $\exists \mathbf{E} \in \Lambda^0.\
  \forall M \in \Lambda^0.\
  \mathbf{E} \, \ulcorner M \urcorner \rightarrow^\ast M$
\end{thm}

\noindent It is not hard to see that, using Theorem
\ref{thm:lambdainterp}, the SRT implies the FRT for closed terms:
given $f \in \Lambda^0$ we let $F \defeq \lambda y.\ f
(\mathbf{E}\,y)$, so that the SRT applied to $F$ yields a term $u$
such that \[
  u =_\beta F \, \ulcorner u \urcorner
    =_\beta f \, (\mathbf{E}\,\ulcorner u \urcorner)
    =_\beta f \, u
\]
It is not at all evident whether the converse holds. This is
because the SRT is a \emph{first-order theorem} that is about
diagonalisation, G\"odel numbers and source code, whereas the FRT
really is about \emph{higher types}: see the discussion in
\cite[\S 2]{Kavvos2017}.

Hence, in the presence of intensional operations, the SRT affords
us with a much stronger kind of recursion. In fact, it allows for
a certain kind of \emph{computational reflection}, or
\emph{reflective programming}, of the same kind envisaged by Brian
Cantwell Smith \cite{Smith1984}. But the programme of Smith's
\emph{reflective tower} involved a rather mysterious construction
with unclear semantics \cite{Friedman1984,Wand1988,Danvy1988},
eventually leading to a theorem that---even in the presence of a
mild reflective construct, the so-called
\textbf{fexpr}---observational equivalence of programs collapses
to $\alpha$-conversion: see Wand \cite{Wand1998}. Similar forays
have also been attempted by the partial evaluation community: see
\cite{Hansen1989,Jones1992,Jones2013}.

We will use modalities to stop intension from flowing back into
extension, so that the aforementioned theorem in
\cite{Wand1998}---which requires unrestricted quoting---will not
apply. We will achieve reflection by internalising the SRT.
Suppose that our terms are typed, and that $u : A$.  Suppose as
well that there is a type constructor $\Box$, so that $\Box A$
means `code of type $A$.' Then certainly $\ulcorner u \urcorner : \Box
A$, and $f$ is forced to have type $\Box A \rightarrow A$. A
logical reading of the SRT is then the following: for every $f :
\Box A \rightarrow A$, there exists a $u : A$ such that $u = f \,
\ulcorner u \urcorner$. This corresponds to \emph{L\"ob's rule} from
\emph{provability logic} \cite{Boolos1994}, namely \[
  \begin{prooftree}
    \Box A \rightarrow A
      \justifies
    A
  \end{prooftree}
\] which is equivalent to adding the G\"odel-L\"ob axiom to the
logic. In fact, the punchline of this paper is that \emph{the type
of the Second Recursion Theorem is the G\"odel-L\"ob axiom of
provability logic}.

To our knowledge, this paper presents the first sound, type-safe
attempt at reflective programming.

\subsection{Prospectus}

In \S\ref{sec:ipcfintro} we will introduce the syntax of iPCF, and
in \S\ref{sec:ipcfmeta} we will show that it satisfies basic
metatheoretic properties. Following that, in section
\S\ref{sec:ipcfconfl} we will add intensional operations to iPCF.
By proving that the resulting notion of reduction is confluent, we
will obtain consistency for the system. We then look at the
computational behaviour of some important terms in
\S\ref{sec:ipcfterms}, and conclude with two key examples of the
new powerful features of our language in \S\ref{sec:ipcfexamples}.

\section{Introducing Intensional PCF}
  \label{sec:ipcfintro}

Intensional PCF (iPCF) is a typed $\lambda$-calculus with modal
types. As discussed before, the modal types work in our favour by
separating intension from extension, so that the latter does not
leak into the former. Given the logical flavour of our previous
work on categorical models of intensionality \cite{Kavvos2017a},
we shall model the types of iPCF after the \emph{constructive
modal logic \textsf{S4}}, in the dual-context style pioneered by
Pfenning and Davies \cite{Davies2001,Davies2001a}. Let us seize
this opportunity to remark that (a) there are also other ways to
capture \textsf{S4}, for which see the survey \cite{Kavvos2016b},
and that (b) dual-context formulations are not by any means
limited to \textsf{S4}: they began in the context of
\emph{intuitionistic linear logic} \cite{Barber1996}, but have
recently been shown to also encompass other modal logics: see
\cite{Kavvos2017b}.

iPCF is \emph{not} related to the language \textsf{Mini-ML} that
is introduced by \cite{Davies2001a}: that is a call-by-value,
ML-like language, with ordinary call-by-value fixed points. In
contrast, ours is a call-by-name language with a new kind of fixed
point, namely intensional fixed points. These fixed points will
afford the programmer the full power of \emph{intensional
recursion}.  In logical terms they correspond to throwing the
G\"odel-L\"ob axiom $\Box(\Box A \rightarrow A) \rightarrow \Box
A$ into \textsf{S4}. Modal logicians might object to this, as, in
conjunction with the \textsf{T} axiom $\Box A \rightarrow A$, it
will make every type inhabited. We remind them that a similar
situation occurs in PCF, where the $\mathbf{Y}_A : (A \rightarrow
A) \rightarrow A$ combinator allows one to write a term
$\mathbf{Y}_A(\lambda x : A.\ x)$ at every type $A$. As in the
study of PCF, we care less about the logic and more about the
underlying computation: \emph{it is the terms that matter, and the
types are only there to stop basic type errors from happening}.

The syntax and the typing rules of iPCF may be found in Figure
\ref{fig:ipcf}. These are largely the same as Pfenning and Davies'
\textsf{S4}, save the addition of some constants (drawn from PCF),
and a rule for intensional recursion. The introduction rule for
the modality restricts terms under a $\ibox{(-)}$ to those
containing only modal variables, i.e. variables carrying only
intensions or code, but never `live values:' \[
  \begin{prooftree}
    \ctxt{\Delta}{\cdot} \vdash M : A
      \justifies
    \ctxt{\Delta}{\Gamma} \vdash \ibox{M} : \Box A
  \end{prooftree}
\] There is also a rule for intensional recursion: \[
  \begin{prooftree}
    \ctxt{\Delta}{z : \Box A} \vdash M : A
      \justifies
    \ctxt{\Delta}{\Gamma} \vdash \fixlob{z}{M} : A
  \end{prooftree}
\]
This will be coupled with the reduction $\fixlob{z}{M} \red{}
M[\ibox{(\fixlob{z}{M})}/z]$. This rule is actually just
\emph{L\"ob's rule} with a modal context, and including it in the
Hilbert system of a (classical or intuitionistic) modal logic is
equivalent to including the G\"odel-L\"ob axiom: see
\cite{Boolos1994} and \cite{Ursini1979a}.  Finally, let us record
the fact that erasing the modality from the types appearing in
either L\"ob's rule or the G\"odel-L\"ob axiom yields the type of
$\mathbf{Y}_A : (A \rightarrow A) \rightarrow A$, as a rule in the
first case, or axiomatically internalised as a constant in the
second (both variants exist in the literature: see
\cite{Gunter1992} and \cite{Mitchell1996}). A similar observation
for a stronger form of the L\"ob axiom underlies the stream of
work on \emph{guarded recursion} \cite{Nakano2000,Birkedal2012};
we recommend the survey \cite{Litak2014} for a broad coverage of
constructive modalities with a provability-like flavour.

\begin{figure}
  \caption{Syntax and Typing Rules for Intensional PCF}
  \label{fig:ipcf}
  \begin{framed}
    \input{ipcfdefn}

  \end{framed}
\end{figure}

\section{Metatheory}
  \label{sec:ipcfmeta}

iPCF satisfies the expected basic results: structural and cut
rules are admissible. This is no surprise given its origin in the
well-behaved Davies-Pfenning calculus. We assume the typical
conventions for $\lambda$-calculi: terms are identified up to
$\alpha$-equivalence, for which we write $\equiv$, and
substitution $[\cdot / \cdot]$ is defined in the ordinary,
capture-avoiding manner. Bear in mind that we consider occurrences
of $u$ in $N$ to be bound in $\letbox{u}{M}{N}$.  Contexts
$\Gamma$, $\Delta$ are lists of type assignments $x : A$.
Furthermore, we shall assume that whenever we write a judgement
like $\ctxt{\Delta}{\Gamma} \vdash M : A$, then $\Delta$ and
$\Gamma$ are \emph{disjoint}, in the sense that $\vars{\Delta}
\cap \vars{\Gamma} = \emptyset$, where $\vars{x_1 : A_1, \dots,
x_n : A_n} \defeq \{x_1, \dots, x_n\}$.  We write $\Gamma, \Gamma'$
for the concatenation of disjoint contexts. Finally, we sometimes
write $\vdash M : A$ whenever $\ctxt{\cdot}{\cdot} \vdash M : A$.

\begin{thm}[Structural \& Cut]
  \label{thm:scut}
  The following rules are admissible in iPCF:
  \begin{multicols}{2}
  \begin{enumerate}
    \item (Weakening) \[
      \begin{prooftree}
        \ctxt{\Delta}{\Gamma, \Gamma'} \vdash M : A
          \justifies
        \ctxt{\Delta}{\Gamma, x : A, \Gamma'} \vdash M : A
      \end{prooftree}
    \]
    \item (Exchange) \[
      \begin{prooftree}
        \ctxt{\Delta}{\Gamma, x : A, y : B, \Gamma'} \vdash M : C
          \justifies
        \ctxt{\Delta}{\Gamma, y : B, x : A, \Gamma'} \vdash M : C
      \end{prooftree}
    \]
    \item (Contraction) \[
      \begin{prooftree}
        \ctxt{\Delta}{\Gamma, x : A, y : A, \Gamma'} \vdash M : A
          \justifies
        \ctxt{\Delta}{\Gamma, w : A, \Gamma'} \vdash M[w, w/x, y] : A
      \end{prooftree}
    \]
    \item (Cut) \[
      \begin{prooftree}
        \ctxt{\Delta}{\Gamma} \vdash N : A
          \qquad
        \ctxt{\Delta}{\Gamma, x : A, \Gamma'} \vdash M : A
          \justifies
        \ctxt{\Delta}{\Gamma, \Gamma'} \vdash M[N/x] : A
      \end{prooftree}
    \]
  \end{enumerate}
  \end{multicols}
\end{thm}

\opt{th}{
  \begin{proof}
    All by induction on the typing derivation of $M$. Verified in
    the proof assistant \textsc{Agda}: see Appendix
    \ref{sec:iPCF.agda}.
  \end{proof}
}

\begin{thm}[Modal Structural \& Cut]
  \label{thm:modalstruct}
  The following rules are admissible:
  \begin{multicols}{2}
  \begin{enumerate}
    \item (Modal Weakening) \[
      \begin{prooftree}
        \ctxt{\Delta, \Delta' }{\Gamma} \vdash M : C
          \justifies
        \ctxt{\Delta, u : A, \Delta'}{\Gamma} \vdash M : C
      \end{prooftree}
    \]
    \item (Modal Exchange) \[
      \begin{prooftree}
        \ctxt{\Delta, x : A, y : B, \Delta'}{\Gamma} \vdash M : C
          \justifies
        \ctxt{\Delta, y : B, x : A, \Delta'}{\Gamma} \vdash M : C
      \end{prooftree}
    \]
    \item (Modal Contraction) \[
      \begin{prooftree}
        \ctxt{\Delta, x : A, y : A, \Delta'}{\Gamma} \vdash M : C
          \justifies
        \ctxt{\Delta, w : A, \Delta'}{\Gamma} \vdash M[w, w/x, y] : C
      \end{prooftree}
    \]
    \item (Modal Cut) \[
      \begin{prooftree}
        \ctxt{\Delta}{\cdot} \vdash N :  A
          \quad
        \ctxt{\Delta, u : A, \Delta'}{\Gamma} \vdash M : C
          \justifies
        \ctxt{\Delta, \Delta'}{\Gamma} \vdash M[N/u] : C
      \end{prooftree}
    \]
  \end{enumerate}
  \end{multicols}
\end{thm}

\subsection{Free variables}

In this section we prove a theorem regarding the occurrences of
free variables in well-typed terms of iPCF. It turns out that, if
a variable occurs free under a $\ibox{(-)}$ construct, then it has
to be in the modal context. This is the property that enforces
that \emph{intensions can only depend on intensions}.

\begin{defn}[Free variables] \hfill
  \begin{enumerate}
    \item The \emph{free variables} $\fv{M}$ of a term
    $M$ are defined by induction on the structure of the term:
    \begin{align*}
      \fv{x}  &\defeq \{x\} &
      \fv{MN} &\defeq \fv{M} \cup \fv{N} \\
      \fv{\lambda x : A.\ M} &\defeq \fv{M} - \{x\} &
      \fv{\ibox{M}} &\defeq \fv{M} \\
      \fv{\fixlob{z}{M}} &\defeq \fv{M} - \{z\}
    \end{align*}
    as well as \[
      \fv{\letbox{u}{M}{N}} \defeq \fv{M} \cup \left(\fv{N} - \{u\}\right)
    \] and $\fv{c} \defeq \emptyset$ for any constant $c$.

    \item The \emph{unboxed free variables} $\ufv{M}$ of a term
    are those that do \emph{not} occur under the scope of a
    $\ibox{(-)}$ or $\fixlob{z}{(-)}$ construct. They are formally
    defined by replacing the following clauses in the definition
    of $\fv{-}$:
    \begin{align*}
      \ufv{\ibox{M}} &\defeq \emptyset &
      \ufv{\fixlob{z}{M}} &\defeq \emptyset
    \end{align*}

    \item The \emph{boxed free variables} $\bfv{M}$ of a term $M$
    are those that \emph{do} occur under the scope of a
    $\ibox{(-)}$ construct. They are formally defined by replacing
    the following clauses in the definition of $\fv{-}$:
    \begin{align*}
      \bfv{x} &\defeq \emptyset &
      \bfv{\ibox{M}} &\defeq \fv{M} \\
      \bfv{\fixlob{z}{M}} &\defeq \fv{M} - \{z\}
    \end{align*}
  \end{enumerate}
\end{defn}

\begin{thm}[Free variables] \hfill
  \label{thm:freevar}
  \begin{enumerate}

    \item For every term $M$, $\fv{M} = \ufv{M} \cup \bfv{M}$.

    \item If and $\ctxt{\Delta}{\Gamma} \vdash M : A$, then
      \begin{align*}
	\ufv{M} &\subseteq \vars{\Gamma} \cup \vars{\Delta} \\
\bfv{M} &\subseteq \vars{\Delta}
     \end{align*}

  \end{enumerate}
\end{thm}

\begin{proof} \hfill
  \begin{enumerate}
    \item
      Trivial induction on $M$.
    \item
      By induction on the derivation of $\ctxt{\Delta}{\Gamma}
      \vdash M : A$.
  \end{enumerate}
\end{proof}

\section{Consistency of Intensional Operations}
  \label{sec:ipcfconfl}

In this section we shall prove that the modal types of iPCF enable
us to consistently add intensional operations on the modal types.
These are \emph{non-functional operations on terms} which are not
ordinarily definable because they violate equality.  All we have
to do is assume them as constants at modal types, define their
behaviour by introducing a notion of reduction, and then
prove that the compatible closure of this notion of reduction is
confluent. A known corollary of confluence is that the equational
theory induced by the reduction is \emph{consistent}, i.e. does
not equate all terms.

There is a caveat involving extension flowing into intension. That
is: we need to exclude from consideration terms where a variable
bound by a $\lambda$ occurs under the scope of a $\ibox{(-)}$
construct. These will never be well-typed, but---since we discuss
types and reduction orthogonally---we also need to explicitly
exclude them here too.

\subsection{Adding intensionality}

\opt{ipcf}{Davies and Pfenning} \cite{Davies2001} suggested that
the $\Box$ modality can be used to signify intensionality. In
fact, in \cite{Davies2001,Davies2001a} they had prevented
reductions from happening under $\ibox{(-)}$ construct, `` [...]
since this would violate its intensional nature.'' But the truth
is that neither of these presentations included any genuinely
non-functional operations at modal types, and hence their only use
was for homogeneous staged metaprogramming. Adding intensional,
non-functional operations is a more difficult task. Intensional
operations are dependent on \emph{descriptions} and
\emph{intensions} rather than \emph{values} and \emph{extensions}.
Hence, unlike reduction and evaluation, they cannot be blind to
substitution. This is something that quickly came to light as soon
as Nanevski \cite{Nanevski2002} attempted to extend the system of
Davies and Pfenning to allow `intensional code analysis' using
nominal techniques.

A similar task was also recently taken up by Gabbay and Nanevski
\cite{Gabbay2013}, who attempted to add a construct
$\textsf{is-app}$ to the system of Davies and Pfenning, along with
the reduction rules \begin{align*}
  \textsf{is-app}\ (\ibox{PQ}) &\red{} \textsf{true} \\
  \textsf{is-app}\ (\ibox{M})  &\red{} \textsf{false}
    \qquad \text{if $M$ is not of the form $PQ$}
\end{align*} The function computed by $\textsf{is-app}$ is truly
intensional, as it depends solely on the syntactic structure of
its argument: it merely checks if it syntactically is an
application or not. As such, it can be considered a
\emph{criterion of intensionality}, albeit an extreme one: its
definability conclusively confirms the presence of computation up
to syntax.

Gabbay and Nanevski tried to justify the inclusion of
$\textsf{is-app}$ by producing denotational semantics for modal
types in which the semantic domain $\sem{\Box A}{}$ directly
involves the actual closed terms of type $\Box A$. However,
something seems to have gone wrong with substitution. In fact, we
believe that their proof of soundness is wrong: it is not hard to
see that their semantics is not stable under the second of these
two reductions: take $M$ to be $u$, and let the semantic
environment map $u$ to an application $PQ$, and then notice that
this leads to $\sem{\textsf{true}}{} = \sem{\textsf{false}}{}$. We
can also see this in the fact that their notion of reduction is
\emph{not confluent}. Here is the relevant counterexample: we can
reduce like this: \[
  \letbox{u}{\ibox{(PQ)}}{\textsf{is-app}\ (\ibox{u})}
    \red{}
  \textsf{is-app}\ (\ibox{PQ})
    \red{}
  \textsf{true}
\] But we could have also reduced like that: \[
  \letbox{u}{\ibox{(PQ)}}{\textsf{is-app}\ (\ibox{u})}
    \red{}
  \letbox{u}{\ibox{(PQ)}}{\textsf{false}}
    \red{}
  \textsf{false}
\] This example is easy to find if one tries to plough through a
proof of confluence: it is very clearly \emph{not} the case that
$M \red{} N$ implies $M[P/u] \red{} N[P/u]$ if $u$ is under a
$\ibox{(-)}$, exactly because of the presence of intensional
operations such as $\textsf{is-app}$.

Perhaps the following idea is more workable: let us limit
intensional operations to a chosen set of functions $f :
\mathcal{T}(A) \rightarrow \mathcal{T}(B)$ from terms of type $A$
to terms of type $B$, and then represent them in the language by a
constant $\tilde f$, such that $\tilde f(\ibox{M}) \red{}
\ibox{f(M)}$. This set of functions would then be chosen so that
they satisfy some sanity conditions. Since we want to have a
\textsf{let} construct that allows us to substitute code for modal
variables, the following general situation will occur: if $N
\red{} N'$, we have
%  \[
%    \begin{tikzcd}
%      & \letbox{u}{\ibox{M}}{N}
%	  \arrow[dr]
%	  \arrow[dl]
%      &  \\
%      N[M/u]
%      &
%      &
%      \letbox{u}{\ibox{M}}{N'}
%	\arrow[d] \\
%      &
%      &
%      N'[M/u]
%    \end{tikzcd}
%  \]
  \[
    \letbox{u}{\ibox{M}}{N}
      \red{}
    N[M/u]
  \] but also \[
    \letbox{u}{\ibox{M}}{N}
      \red{}
    \letbox{u}{\ibox{M}}{N'}
      \red{}
    N'[M/u]
  \]

\noindent Thus, in order to have confluence, we need $N[M/u]
\red{} N'[M/u]$. This will only be the case for reductions of the
form $\tilde f(\ibox{M}) \rightarrow \ibox{f(M)}$ if \[
  f(N[M/u]) \equiv f(N)[M/u]
\] i.e. if $f$ is \emph{substitutive}. But then a simple
naturality argument gives that $f(N) \equiv f(u[N/u]) \equiv
f(u)[N/u]$, and hence $\tilde f$ is already definable by \[
  \lambda x : \Box A.\ \letbox{u}{x}{\ibox{f(u)}}
\] so such a `substitutive' function is not intensional after all.

In fact, the only truly intensional operations we can add to our
calculus will be those acting on \emph{closed} terms. We will see
that this circumvents the problems that arise when intensionality
interacts with substitution. Hence, we will limit intensional
operations to the following set:

\begin{defn}[Intensional operations]
  Let $\mathcal{T}_0(A)$ be the set of ($\alpha$-equivalence
  classes of) closed terms $M$ such that $\ctxt{\cdot}{\cdot}
  \vdash M : A$.  Then, the set of \emph{intensional operations},
  $\mathcal{F}(A, B)$, is defined to be the set of all functions
  $f : \mathcal{T}_0(A) \rightarrow \mathcal{T}_0(B)$.
\end{defn}
We will include all of these intensional operations $f :
\mathcal{T}_0(A) \rightarrow \mathcal{T}_0(B)$ in our calculus as
constants: \[
  \begin{prooftree}
    \justifies
      \ctxt{\Delta}{\Gamma} \vdash \tilde f : \Box A \rightarrow \Box B
  \end{prooftree}
\] with reduction rule $\tilde f(\ibox{M}) \rightarrow
\ibox{f(M)}$, under the proviso that $M$ is closed.  Of course,
these also includes operations on terms that \emph{might not be
computable}. However, we are interested in proving consistency of
intensional operations in the most general setting. The questions
of which intensional operations are computable, and which
primitives or mechanisms can and should be used to express them,
are beyond the scope of this paper, and largely still open.

\subsection{Reduction and Confluence}

\begin{figure}
  \caption{Reduction for Intensional PCF}
  \label{fig:ipcfbeta}
  \begin{framed}
    \input{ipcfbeta.tex}
  \end{framed}
\end{figure}

\opt{th}{
\begin{figure}
  \caption{Equational Theory for Intensional PCF}
  \label{fig:ipcfeq}
  \input{ipcfeq.tex}
\end{figure}
}

We introduce a notion of reduction for iPCF, which we present in
Figure \ref{fig:ipcfbeta}. Unlike many studies of PCF-inspired
languages, we do not consider a reduction strategy but ordinary
`non-deterministic' $\beta$-reduction. We do so because are trying
to show consistency of the induced equational theory.

The equational theory induced by this notion of reduction is a
symmetric version of it, annotated with types. It is easy to write
down, so we omit it. Note the fact that, like the calculus of
Davies and Pfenning, we do \emph{not} include the following
congruence rule for the modality:
\[
  \begin{prooftree}
      \ctxt{\Delta}{\cdot} \vdash M = N : A
    \justifies
      \ctxt{\Delta}{\Gamma} \vdash \ibox{M} = \ibox{N} : \Box A
    \using
      {(\Box\textsf{cong})}
  \end{prooftree}
\] In fact, the very absence of this rule is what will allow modal
types to become intensional. Otherwise, the only new rules are
intensional recursion, embodied by the rule $(\Box \mathsf{fix})$,
and intensional operations, exemplified by the rule
$(\Box\mathsf{int})$.

We note that it seems perfectly reasonable to think that we should
allow reductions under $\textsf{fix}$, i.e. admit the rule \[
  \begin{prooftree}
      M \red{} N
    \justifies
      \fixlob{z}{M} \red{} \fixlob{z}{N}
  \end{prooftree}
\] as $M$ and $N$ are expected to be of type $A$, which need not
be modal. However, the reduction $\fixlob{z}{M} \red{}
M[\ibox{(\fixlob{z}{M})}/z]$ `freezes' $M$ under an occurrence of
$\ibox{(-)}$, so that no further reductions can take place within
it. Thus, the above rule would violate the intensional nature of
boxes. We were likewise compelled to define $\ufv{\fixlob{z}{M}}
\defeq \emptyset$ in the previous section: we should already
consider $M$ to be intensional, or under a box.

We can now show that

\begin{thm}
  \label{thm:conf}
  The reduction relation $\red{}$ is confluent.
\end{thm}

The easiest route to that theorem is to use a proof like that in
\cite{Kavvos2017b}, i.e. the method of \emph{parallel reduction}.
This kind of proof was originally discovered by Tait and
Martin-L\"of, and is nicely documented in \cite{Takahashi1995}.
Because of the intensional nature of our $\ibox{(-)}$ constructs,
ours will be more nuanced and fiddly.  The proof can of course be
skipped on a first reading.

\paragraph{Proof of confluence}

\input{ipcfconflproof}

\section{Some important terms}
  \label{sec:ipcfterms}

Let us look at the kinds of terms we can write in iPCF.

\begin{description}

\item[From the axioms of \textsf{S4}]

First, we can write a term corresponding to axiom \textsf{K}, the
\emph{normality axiom} of modal logics: \[
  \mathsf{ax_K} \defeq
    \lambda f : \Box (A \rightarrow B). \;
    \lambda x : \Box A . \;
      \letbox{g}{f}{\letbox{y}{x}{\ibox{(g\,y)}}}
\]
Then $\vdash \mathsf{ax_K} : \Box(A \rightarrow B) \rightarrow
(\Box A \rightarrow \Box B)$. An intensional reading of this is
the following: any function given as code can be transformed into
an \emph{effective operation} that maps code of type $A$ to code
of type $B$.

The rest of the axioms correspond to evaluating and quoting.
Axiom \textsf{T} takes code to value, or intension to extension:
\[
  \vdash \mathsf{eval}_A \defeq
    \lambda x : \Box A. \; \letbox{y}{x}{y}
      : \Box A \rightarrow A
\] and axiom \textsf{4} quotes code into code-for-code: \[
  \vdash \mathsf{quote}_A \defeq
    \lambda x : \Box A. \;
      \letbox{y}{x}{\ibox{\left(\ibox{y}\right)}}
    : \Box A \rightarrow \Box \Box A
\]

\item[The G\"odel-L\"ob axiom: intensional fixed points]

Since $(\Box\textsf{fix})$ is L\"ob's rule, we expect to be able
to write down a term corresponding to the G\"odel-L\"ob axiom of
provability logic. We can, and it is an \emph{intensional
fixed-point combinator}: \[
  \mathbb{Y}_A \defeq
    \lambda x : \Box (\Box A \rightarrow A). \;
      \letbox{f}{x}{
          \ibox{\left(\fixlob{z}{f\,z}\right)}
      }
\] and $\vdash \mathbb{Y}_A : \Box(\Box A \rightarrow A)
\rightarrow \Box A$. We observe that \[
  \mathbb{Y}_A(\ibox{M})
    \redt{}
  \ibox{\left(\fixlob{z}{(M\, z)}\right)}
\]

\item[Undefined]

The combination of $\textsf{eval}$ and intensional fixed points
leads to non-termination, in a style reminiscent of the term
$(\lambda x.\,xx)(\lambda x.\,xx)$ of the untyped
$\lambda$-calculus. Let \[
  \Omega_A \defeq \fixlob{z}{(\textsf{eval}_A\,z)}
\] Then $\vdash \Omega_A : A$, and \[
  \Omega_A
    \red{}\
  \textsf{eval}_A\;\left(\ibox{\Omega_A}\right)\
    \redt{}\
  \Omega_A
\]

\item[Extensional Fixed Points]

Perhaps surprisingly, the ordinary PCF $\mathbf{Y}$ combinator is
also definable in the iPCF.  Let \[
  \mathbf{Y}_A \defeq
    \fixlob{z}{
      \lambda f : A \rightarrow A.\
        f (\textsf{eval}\ z\  f)
    }
\] Then $\vdash \mathbf{Y}_A : (A \rightarrow A) \rightarrow A$,
so that \begin{align*}
  \mathbf{Y}_A
    \redt{}\
  &\lambda f : A \rightarrow A.\ f (\textsf{eval}\ (\ibox{\mathbf{Y}_A})\ f)) \\
    \redt{}\
  &\lambda f : A \rightarrow A.\ f (\mathbf{Y}_A\, f)
\end{align*}

\opt{th}{Notice that, in this term, the modal variable $z$ occurs
free under a $\lambda$-abstraction. This will prove important in
\S\ref{chap:intsem2}, where it will be prohibited.}

\end{description}

\section{Two intensional examples}
  \label{sec:ipcfexamples}

No discussion of an intensional language with intensional
recursion would be complete without examples that use these two
novel features. Our first example uses intensionality, albeit in a
`extensional' way, and is drawn from the study of PCF and issues
related to sequential vs. parallel (but not concurrent)
computation. Our second example uses intensional recursion, so it
is slightly more adventurous: it is a computer virus.

\subsection{`Parallel or' by dovetailing}

In \cite{Plotkin1977} Gordon Plotkin proved the following
theorem: there is no term $\textsf{por} : \textsf{Bool}
\rightarrow \textsf{Bool} \rightarrow \textsf{Bool}$ of PCF such
that $\textsf{por}\ \textsf{true}\ M \twoheadrightarrow_\beta
\textsf{true}$ and  $\textsf{por}\ M\ \textsf{true}
\twoheadrightarrow_\beta \textsf{true}$ for any $\vdash M :
\textsf{Bool}$, whilst $\textsf{por}\ \textsf{false}\
\textsf{false} \twoheadrightarrow_\beta \textsf{false}$.
Intuitively, the problem is that $\textsf{por}$ has to first
examine one of its two arguments, and this can be troublesome if
that argument is non-terminating. It follows that the
\emph{parallel or} function is not definable in PCF. In order to
regain the property of so-called \emph{full abstraction} for the
\emph{Scott model} of PCF, a constant denoting this function has
to be manually added to PCF, and endowed with the above rather
clunky operational semantics. See
\cite{Plotkin1977,Gunter1992,Mitchell1996,Streicher2006}.

However, the parallel or function is a computable \emph{partial
recursive functional} \cite{Streicher2006,Longley2015}. The way to
prove that is intuitively the following: given two closed terms
$M, N : \textsf{Bool}$, take turns in $\beta$-reducing each one
for a one step: this is called \emph{dovetailing}. If at any point
one of the two terms reduces to \textsf{true}, then output
\textsf{true}. But if at any point both reduce to \textsf{false},
then output \textsf{false}.

This procedure is not definable in PCF because a candidate term
$\textsf{por}$ does not have access to a code for its argument,
but can only inspect its value. However, in iPCF we can use the
modality to obtain access to code, and intensional operations to
implement reduction. Suppose we pick a reduction strategy
$\red{}_r$. Then, let us include a constant $\textsf{tick} :
\Box\textsf{Bool} \rightarrow \Box\textsf{Bool}$ that implements
one step of this reduction strategy on closed terms: \[
  \begin{prooftree}
    M \red_r N, \text{ $M, N$ closed}
      \justifies
    \textsf{tick}\ (\ibox{M}) \red{} \ibox{N}
  \end{prooftree}
\] Also, let us include a constant $\textsf{done?} :
\Box\textsf{Bool} \rightarrow \textsf{Bool}$, which tells us if a
closed term under a box is a normal form: \[
  \begin{prooftree}
    \text{$M$ closed, normal}
      \justifies
    \textsf{done?}\ (\ibox{M}) \red{} \textsf{true}
  \end{prooftree}
    \qquad
  \begin{prooftree}
    \text{$M$ closed, not normal}
      \justifies
    \textsf{done?}\ (\ibox{M}) \red{} \textsf{false}
  \end{prooftree}
\] These two can be subsumed under our previous scheme for
introducing intensional operations. The above argument is now
implemented by the following term: \begin{align*}
  \textsf{por} :\equiv\
    \mathbf{Y}(
      \lambda &\textsf{por}. \;
      \lambda x : \Box \textsf{Bool}. \; \lambda y : \Box \textsf{Bool}. \\
    &\supset_\textsf{Bool}
      (\textsf{done?} \; x)\
      &&(\textsf{lor} \; (\textsf{eval} \; x) (\textsf{eval} \; y)) \\
    & &&(\supset_\textsf{Bool} (\textsf{done?} \; y)\
      &&(\textsf{ror} \; (\textsf{eval} \; x) (\textsf{eval} \; y)) \\
    & && &&(\textsf{por} \; (\textsf{tick} \; x)
                            (\textsf{tick} \; y)))
\end{align*} where $\textsf{lor}, \textsf{ror} : \textsf{Bool}
\rightarrow \textsf{Bool} \rightarrow \textsf{Bool}$ are terms
defining the left-strict and right-strict versions of the `or'
connective respectively. Notice that the type of this term is
$\Box\textsf{Bool} \rightarrow \Box\textsf{Bool} \rightarrow
\textsf{Bool}$: we require \emph{intensional access} to the terms
of boolean type in order to define this function!

\subsection{A computer virus}

\emph{Abstract computer virology} is the study of formalisms that
model computer viruses. There are many ways to formalise viruses.
We will use the model of \opt{ipcf}{Adleman} \cite{Adleman1990},
where files can be interpreted either as data, or as functions. We
introduce a data type $F$ of files, and two constants \[
  \textsf{in} : \Box (F \rightarrow F) \rightarrow F
    \quad \text{ and} \quad
  \textsf{out} : F \rightarrow \Box (F \rightarrow F)
\] If $F$ is a file, then $\textsf{out}\ F$ is that file
interpreted as a program, and similarly for $\textsf{in}$. We
ask that $\textsf{out}\ (\textsf{in}\ M) \red{} M$,
making $\Box (F \rightarrow F)$ a retract of
$F$.\opt{th}{\footnote{Actually, in \S\ref{sec:buildipwps} and
\S\ref{sec:asmipcf} we will see it is very easy to construct
examples for the apparently more natural situation where
$\textsf{in} : \Box(\Box F \rightarrow F) \rightarrow F$,
$\textsf{out}: F \rightarrow (\Box F \rightarrow F)$, and
$\textsf{out}\ (\textsf{in}\ M) \red{}\ \textsf{eval}_{F
\rightarrow F}\ M $. Nevertheless, our setup is slightly more
well-adapted to virology.}} This might seem the same as the
situation where $F \rightarrow F$ is a retract of $F$, which
yields models of the (untyped) $\lambda$-calculus, and is not
trivial to construct \cite[\S 5.4]{Barendregt1984}. However, in
our case it is not nearly as worrying: $\Box(F \rightarrow F)$ is
populated by programs and codes, not by actual functions. Under
this interpretation, the pair $(\textsf{in}, \textsf{out})$
corresponds to a kind of G\"odel numbering---especially if $F$ is
$\mathbb{N}$.

In Adleman's model, a \emph{virus} is given by its infected form,
which either \emph{injures}, \emph{infects}, or \emph{imitates}
other programs. The details are unimportant in the present
discussion, save from the fact that the virus needs to have access
to code that it can use to infect other executables.  One can
hence construct such a virus from its \emph{infection routine}, by
using Kleene's SRT. Let us model it by a term \[
  \vdash \textsf{infect} : \Box (F \rightarrow F)
    \rightarrow F \rightarrow F
\] which accepts a piece of viral code and an executable file, and
it returns either the file itself, or a version infected with the
viral code. We can then define a term \[
  \vdash \textsf{virus}
    \defeq \fixlob{z}{(\textsf{infect}\ z)} :
      F \rightarrow F
\] so that \[
  \textsf{virus}\
    \redt{}\
  \textsf{infect}\ (\ibox{\textsf{virus}})
\] which is a program that is ready to infect its input with its
own code.

\section{Conclusion}

We have achieved the desideratum of an intensional programming
calculus with intensional recursion. There are two main questions
that result from this development.

First, does there exist a good set of \emph{intensional
primitives} from which all others are definable? Is there perhaps
\emph{more than one such set}, hence providing us with a choice of
programming primitives? Previous attempts aiming to answer this
question include those of \cite{Schurmann2001,Nanevski2002}.

Second, what is the exact kind of programming power that we have
unleashed?  Does it lead to interesting programs that we have not
been able to write before? We have outlined some speculative
applications for intensional recursion in \cite[\S\S
1--2]{Kavvos2017}. Is iPCF a useful tool when it comes to
attacking these?

%%% ARTICLE ENDS HERE

% Acknowledgments

\section*{Acknowledgements}

The author would like to thank Mario Alvarez-Picallo for their
endless conversations on types and metaprogramming, Neil Jones for
his careful reading and helpful comments, and Samson Abramsky for
suggesting the topic of intensionality. This work was supported by
the EPSRC (award reference 1354534).

\bibliographystyle{entcs}
\bibliography{\string~/cs/lib/library.bib}

% appendix
\appendix

\end{document}

%% file: ipcfdefn.tex
\begin{align*}
  \textbf{Ground Types} \quad &
    G & ::=\quad &\textsf{Nat} \;|\; \textsf{Bool}
   \\ \\
  \textbf{Types} \quad &
    A, B & ::=\quad &G \;|\; A \rightarrow B \;|\; \Box A
   \\ \\
  \textbf{Terms} \quad &
    M, N & ::=\quad &x
      \;|\; \lambda x{:}A.\ M
      \;|\; M N
      \;|\; \ibox{M}
      \;|\; \letbox{u}{M}{N} \;| \\
   &      &    &\widehat{n}
      \;|\; \textsf{true}
      \;|\; \textsf{false}
      \;|\; \textsf{succ}
      \;|\; \textsf{pred}
      \;|\; \textsf{zero?}
      \;|\; \supset_G
      \;|\; \fixlob{z}{M}
    \\ \\
%  \textbf{Canonical Forms} \quad &
%    V & ::= \quad &\widehat{n}
%     \;|\; \textsf{true}
%     \;|\; \textsf{false}
%     \;|\; \lambda x{:}A.\ M
%     \;|\; \textsf{box } M
%   \\ \\
  \textbf{Contexts} \quad &
    \Gamma, \Delta & ::=\quad &\cdot \;|\; \Gamma, x: A
\end{align*}

\vfill

\renewcommand{\arraystretch}{3}

% TABLE FOR CONSTANTS
\begin{tabular}{c c}
  %% CONSTANT RULES
  $
    \begin{prooftree}
      \justifies
        \ctxt{\Delta}{\Gamma} \vdash \widehat{n} : \textsf{Nat}
    \end{prooftree}
  $
  
  &

  $   
    \begin{prooftree}
      \justifies
        \ctxt{\Delta}{\Gamma} \vdash b : \textsf{Bool}
      \using
        (b \in \{\textsf{true}, \textsf{false}\})
    \end{prooftree}
  $
  
  \\

  $ 
    \begin{prooftree}
      \justifies
        \ctxt{\Delta}{\Gamma} \vdash \textsf{zero?} : 
          \textsf{Nat} \rightarrow \textsf{Bool}
    \end{prooftree}
  $
  
  &

  $ 
    \begin{prooftree}
      \justifies
        \ctxt{\Delta}{\Gamma}
        \vdash f : \textsf{Nat} \rightarrow \textsf{Nat}
      \using
        (f \in \{\textsf{succ}, \textsf{pred}\})
    \end{prooftree}
  $
  
  \\

  \multicolumn{2}{c}{
  $
    \begin{prooftree}
      \justifies
        \ctxt{\Delta}{\Gamma} \vdash {\supset_G}  : 
        \textsf{Bool} \rightarrow G \rightarrow G \rightarrow G
    \end{prooftree}
  $
  }

  \\

  %% VARIABLE RULES

  $
    \begin{prooftree}
        \justifies
          \ctxt{\Delta}{\Gamma, x{:}A, \Gamma'} \vdash x:A
        \using
          {(\textsf{var})}
    \end{prooftree}
  $

  &

  $
    \begin{prooftree}
        \justifies
      \ctxt{\Delta, u{:} A, \Delta'}{\Gamma} \vdash u:A
        \using
      {(\Box\textsf{var})}
    \end{prooftree}
  $

  \\

  %% IMPLICATION RULES

  $
    \begin{prooftree}
      \ctxt{\Delta}{\Gamma}, x{:}A \vdash M : B
        \justifies
      \ctxt{\Delta}{\Gamma} \vdash \lambda x{:}A. \; M : A \rightarrow B
        \using
      {(\rightarrow\mathcal{I})}
    \end{prooftree}
  $

  &

  $
    \begin{prooftree}
      \ctxt{\Delta}{\Gamma} \vdash M : A \rightarrow B
        \quad
      \ctxt{\Delta}{\Gamma} \vdash N : A
        \justifies
      \ctxt{\Delta}{\Gamma} \vdash M N : B
        \using
      {(\rightarrow\mathcal{E})}
    \end{prooftree}
  $

  \\

  %% BOX RULES

  $
    \begin{prooftree}
      \ctxt{\Delta}{\cdot} \vdash M : A
        \justifies
      \ctxt{\Delta}{\Gamma} \vdash \ibox{M} : \Box A
        \using
      {(\Box\mathcal{I})}
    \end{prooftree}
  $

  &

  $
    \begin{prooftree}
      \ctxt{\Delta}{\Gamma} \vdash M : \Box A
        \quad\quad
      \ctxt{\Delta, u{:}A}{\Gamma} \vdash N : C
        \justifies
      \ctxt{\Delta}{\Gamma} \vdash \letbox{u}{M}{N} : C
        \using
      {(\Box\mathcal{E})}
    \end{prooftree}
  $

  \\

  \multicolumn{2}{c}{
    $
      \begin{prooftree}
        \ctxt{\Delta}{z : \Box A} \vdash M : A
          \justifies
        \ctxt{\Delta}{\Gamma} \vdash \fixlob{z}{M} : \Box A
          \using
        {(\Box\textsf{fix})}
      \end{prooftree}
    $
  }
\end{tabular}

%% file: ipcfbeta.tex
\renewcommand{\arraystretch}{4}

\begin{center}
\begin{tabular}{c c}
  %% BETA
  $
    \begin{prooftree}
        \justifies
      (\lambda x : A.\ M)N \red{} M[N/x]
        \using
      {(\red{}\beta)}
    \end{prooftree}
  $

  &

  $
    \begin{prooftree}
      M \red{} N
        \justifies
    \lambda x : A.\ M \red{} \lambda x : A.\ N
        \using
      {(\textsf{cong}_\lambda)}
    \end{prooftree}
  $

  \\

  %% APPLICATION CONGRUENCE

  $
    \begin{prooftree}
      M \red{} N
        \justifies
      MP \red{} NP
        \using
      {(\textsf{app}_1)}
    \end{prooftree}
  $

  &

  $
    \begin{prooftree}
      P \red{} Q
        \justifies
      MP \red{} MQ
        \using
      {(\textsf{app}_2)}
    \end{prooftree}
  $

  \\

  \multicolumn{2}{c}{
    $
      \begin{prooftree}
          \justifies
        \letbox{u}{\ibox{M}}{N} \red{} N[M/u]
          \using
        {(\Box\beta)}
      \end{prooftree}
    $
  }

  \\

  \multicolumn{2}{c}{
    $
      \begin{prooftree}
          \justifies
        \fixlob{z}{M} \red{} M[\ibox{(\fixlob{z}{M})}/z]
          \using
        {(\Box\textsf{fix})}
      \end{prooftree}
    $
  }
  
  \\

  \multicolumn{2}{c}{
    $
      \begin{prooftree}
        \text{$M$ closed, $M \in \textrm{dom}(f)$}
          \justifies
        \tilde f(\ibox{M}) \red{} \ibox{f(M)}
          \using
        {(\Box\textsf{int})}
      \end{prooftree}
    $
  }

  \\

  %% LET BOX CONGRUENCE

  \multicolumn{2}{c}{
    $
      \begin{prooftree}
        M \red{} N
          \justifies
        \letbox{u}{M}{P} \red{} \letbox{u}{N}{P}
          \using
        {(\textsf{let-cong}_1)}
      \end{prooftree}
    $
  }

  \\

  \multicolumn{2}{c}{
    $
      \begin{prooftree}
        P \red{} Q
          \justifies
        \letbox{u}{M}{P} \red{} \letbox{u}{M}{Q}
          \using
        {(\textsf{let-cong}_2)}
      \end{prooftree}
    $
  }

  \\

  $ 
    \begin{prooftree}
        \justifies
      \textsf{zero?}\ \widehat{0} \red{} \textsf{true}
        \using
      {(\textsf{zero?}_1)}
    \end{prooftree}
  $

  & 

  $
    \begin{prooftree}
        \justifies
      \textsf{zero?}\ \widehat{n+1} \red{} \textsf{false}
        \using
      {(\textsf{zero?}_2)}
    \end{prooftree}
  $

  \\

  $
    \begin{prooftree}
        \justifies
      \textsf{succ}\ \widehat{n} \red{} \widehat{n+1}
        \using
      {(\textsf{succ})}
    \end{prooftree}
  $

  & 

  $
    \begin{prooftree}
        \justifies
      \textsf{pred}\ \widehat{n} \red{} \widehat{n\dotdiv 1}
        \using
      {(\textsf{pred})}
    \end{prooftree}
  $

  \\

  $
    \begin{prooftree}
        \justifies
      \supset_G\ \textsf{true}\ M\ N \red{} M
        \using
      {(\supset_1)}
    \end{prooftree}
  $

  & 

  $
    \begin{prooftree}
        \justifies
      \supset_G\ \textsf{false}\ M\ N \red{} N
        \using
      {(\supset_2)}
    \end{prooftree}
  $
\end{tabular}
\end{center}

%% file: ipcfeq.tex
\renewcommand{\arraystretch}{4}

\begin{tabular}{c c}

%  \textbf{Equivalence Relation} & \\ \\
%
%  \multicolumn{2}{c}{
%    $ \begin{prooftree}
%        \ctxt{\Delta}{\Gamma} \vdash M : A
%          \justifies
%        \ctxt{\Delta}{\Gamma} \vdash M = M : A
%      \end{prooftree} $
%  } \\ \\
%
%  $ \begin{prooftree}
%      \ctxt{\Delta}{\Gamma} \vdash M = N : A
%        \justifies
%      \ctxt{\Delta}{\Gamma} \vdash N = M : A
%    \end{prooftree} $
%  &
%  $ \begin{prooftree}
%      \ctxt{\Delta}{\Gamma} \vdash P = Q : A
%        \qquad
%      \ctxt{\Delta}{\Gamma} \vdash Q = R : A
%        \justifies
%      \ctxt{\Delta}{\Gamma} \vdash P = R : A
%    \end{prooftree} $
%

  \textbf{Function Spaces} & \\

  \multicolumn{2}{c}{
    $
      \begin{prooftree}
        \ctxt{\Delta}{\Gamma} \vdash N : A
          \qquad
        \ctxt{\Delta}{\Gamma, x{:}A, \Gamma'} \vdash M : B
          \justifies
        \ctxt{\Delta}{\Gamma} \vdash (\lambda x{:}A. M)\,N = M[N/x] : B
          \using
        {(\rightarrow\beta)}
      \end{prooftree}
    $
  }

  \\

%  $
%    \begin{prooftree}
%      \ctxt{\Delta}{\Gamma} \vdash M : A \rightarrow B
%        \qquad
%      x \not\in \text{fv}(M)
%        \justifies
%      \ctxt{\Delta}{\Gamma} \vdash M = \lambda x{:}A. Mx : A \rightarrow B
%        \using
%      {(\rightarrow\eta)}
%    \end{prooftree}
%  $

%  $ \begin{prooftree}
%      \ctxt{\Delta}{\Gamma, x{:}A} \vdash M = N : B
%        \justifies
%      \ctxt{\Delta}{\Gamma} \vdash \lambda x{:}A. M = \lambda x{:}A. N :
%      A \rightarrow B
%        \using
%      {(\rightarrow\textsf{cong}_1)}
%    \end{prooftree} $

%  $ \begin{prooftree}
%      \ctxt{\Delta}{\Gamma} \vdash M = P : A \rightarrow B
%        \qquad
%      \ctxt{\Delta}{\Gamma} \vdash N = Q : A
%        \justifies
%      \ctxt{\Delta}{\Gamma} \vdash MN = PQ : B
%        \using
%      {(\rightarrow\textsf{cong}_2)}
%    \end{prooftree} $

  \textbf{Modality} & \\

  \multicolumn{2}{c}{
    $
      \begin{prooftree}
        \ctxt{\Delta}{\cdot} \vdash M : A
          \qquad
        \ctxt{\Delta, u : A}{\Gamma} \vdash N : C
          \justifies
        \ctxt{\Delta}{\Gamma}
            \vdash \letbox{u}{\ibox{M}}{N} = N[M/x] : C
          \using
        {(\Box\beta)}
      \end{prooftree}
    $
  }
    
  \\

%  \multicolumn{2}{c}{
%  $ \begin{prooftree}
%      \ctxt{\Delta}{\Gamma} \vdash M : \Box A
%        \justifies
%      \ctxt{\Delta}{\Gamma} \vdash \letbox{u}{M}{\ibox{u}} = M : \Box A
%        \using
%      {(\Box\eta)}
%    \end{prooftree} $
%  }
%
%  \\

  \multicolumn{2}{c}{
    $
      \begin{prooftree}
        \ctxt{\Delta}{z : \Box A} \vdash M : A
          \justifies
        \ctxt{\Delta}{\Gamma} \vdash 
            \fixlob{z}{M} = M[\ibox{(\fixlob{z}{M})}/z] : A
          \using
          {(\Box\textsf{fix})}
      \end{prooftree}
    $
  }

  \\

  \multicolumn{2}{c}{
    $
      \begin{prooftree}
        \ctxt{\cdot}{\cdot} \vdash M : A
          \quad
        f \in \mathcal{F}(A, B)
          \justifies
        \ctxt{\Delta}{\Gamma} \vdash 
            \tilde f (\ibox{M}) = \ibox{f(M)} : \Box B
          \using
        {(\Box\textsf{int})}
      \end{prooftree}
    $
  }

  \\

  \multicolumn{2}{c}{
    $
      \begin{prooftree}
        \ctxt{\Delta}{\Gamma} \vdash M = N : \Box A
          \qquad
        \ctxt{\Delta}{\Gamma} \vdash P = Q : C
          \justifies
        \ctxt{\Delta}{\Gamma} \vdash
            \letbox{u}{M}{P} = \letbox{u}{N}{Q} : B
          \using
        {(\Box\textsf{let-cong})}
      \end{prooftree}
    $
  }

  \\

  \multicolumn{2}{c}{
    \begin{minipage}{\textwidth}
      \textbf{Remark}. In addition to the above, one should also
      include (a) rules that ensure that equality is an equivalence
      relation, (b) congruence rules for $\lambda$-abstraction and
      application, and (c) rules corresponding to the behaviour of
      constants, as in Figure \ref{fig:ipcfbeta}. 
    \end{minipage}
  }

\end{tabular}

%% file: ipcfconflproof.tex
We will use a variant of the proof in \cite{Kavvos2017b}, i.e.
the method of \emph{parallel reduction}. This kind of proof was
originally discovered by Tait and Martin-L\"of, and is nicely
documented in \cite{Takahashi1995}. Because of the intensional
nature of our $\ibox{(-)}$ constructs, ours will be more nuanced
and fiddly than any in \emph{op. cit.} The method is this: we will
introduce a second notion of reduction, \[
  \redp\ \subseteq \Lambda \times \Lambda
\] which we will `sandwich' between reduction proper and its
transitive closure: \[
  \red{}\ \subseteq\ \redp\ \subseteq\ \redt
\] We will then show that $\redp$ has the diamond property. By the
above inclusions, the transitive closure $\redp^\ast$ of $\redp$
is then equal to $\redt$, and hence $\red$ is Church-Rosser.

In fact, we will follow \cite{Takahashi1995} in doing something
better: we will define for each term $M$ its \emph{complete
development}, $M^\star$. The complete development is intuitively
defined by `unrolling' all the redexes of $M$ at once. We will
then show that if $M \redp N$, then $N \redp M^\star$. $M^\star$
will then suffice to close the diamond: \[
  \begin{tikzcd}
    & M
	\arrow[dr, Rightarrow]
	\arrow[dl, Rightarrow]
    &  \\
    P
      \arrow[dr, Rightarrow, dotted]
    & 
    & Q 
      \arrow[dl, Rightarrow, dotted] \\
    & M^\star
  \end{tikzcd}
\]

\begin{figure}
  \centering
  \caption{Parallel Reduction}
  \begin{framed}
    \input{ipcfparallel}
  \end{framed}
  \label{fig:parallel}
\end{figure}

The parallel reduction $\redp$ is defined in Figure
\ref{fig:parallel}. Instead of the axiom $(\textsf{refl})$ we
would more commonly have an axiom for variables, $x \redp{} x$,
and $M \redp{} M$ would be derivable. However, we do not have a
congruence rule neither for $\ibox{(-)}$ nor for L\"ob's rule, so
that possibility would be precluded. We are thus forced to include
$M \redp{} M$, which slightly complicates the lemmas that follow.

The main lemma that usually underpins the confluence proof is
this: if $M \redp{} N$ and $P \redp{} Q$, $M[P/x] \redp{} N[Q/x]$.
However, this is intuitively wrong: no reductions should happen
under boxes, so this should only hold if we are substituting for a
variable \emph{not} occurring under boxes.  Hence, this lemma
splits into three different ones:
\begin{itemize}
  \item
    $P \redp{} Q$ implies $M[P/x] \redp{} M[Q/x]$, if $x$ does not
    occur under boxes: this is the price to pay for replacing the
    variable axiom with (\textsf{refl}).
  \item
    $M \redp{} N$ implies $M[P/u] \redp{} N[P/u]$, even if $u$ is
    under a box.
  \item
    If $x$ does not occur under boxes, $M \redp{} N$ and $P
    \redp{} Q$ indeed imply $M[P/x] \redp{} N[Q/x]$
\end{itemize} 
\begin{lem}
  \label{lem:substint}
  If $M \redp{} N$ then $M[P/u] \redp{} N[P/u]$.
\end{lem}
\begin{proof}
  By induction on the generation of $M \redp{} N$. Most cases
  trivially follow, or consist of simple invocations of the IH. In
  the case of $(\rightarrow\beta)$, the known substitution lemma
  suffices. Let us look at the cases involving boxes.
  \begin{indproof}
    \indcase{$\Box\beta$}
      Then $M \redp{} N$ is $\letbox{v}{\ibox{R}}{S} \redp{}
      S'[R/v]$ with $S \redp{} S'$. By the IH, we have that
      $S[P/u] \redp{} S'[P/u]$, so \[
	\letbox{v}{\ibox{R[P/u]}}{S[P/u]} \redp S'[P/u][R[P/u]/v]
      \] and this last is $\alpha$-equivalent to $S'[R/v][P/u]$ by
      the substitution lemma.
    \indcase{$\Box\textsf{fix}$}
      A similar application of the substitution lemma.
    \indcase{$\Box\textsf{int}$}
      Then $M \redp{} N$ is $\tilde f(\ibox{Q}) \redp{}
      \ibox{f(Q)}$. Hence \[
        \left(\tilde f(\ibox{Q})\right)[P/u]
          \equiv
        \tilde f(\ibox{Q})
          \redp{}
        \ibox{f(Q)}
          \equiv
        \left(\ibox{f(Q)}\right)[P/u]
    \] simply because both $Q$ and $f(Q)$ are closed.
  \end{indproof}
\end{proof}

\begin{lem} 
  \label{lem:substredp}
  If $P \redp Q$ and $x \not\in \bfv{M}$, then $M[P/x] \redp
  M[Q/x]$.
\end{lem}
\begin{proof}
  By induction on the term $M$. The only non-trivial cases are
  those for $M$ a variable, $\ibox{M'}$ or $\fixlob{z}{M'}$. In
  the first case, depending on which variable $M$ is, use either
  $(\textsf{refl})$, or the assumption $P \redp Q$. In the latter
  two, $(\ibox{M'})[P/x] \equiv \ibox{M'} \equiv (\ibox{M'})[Q/x]$
  as $x$ does not occur under a box, so use $(\textsf{refl})$, and
  similarly for $\fixlob{z}{M'}$.
\end{proof}

\begin{lem} 
  \label{lem:redp}
  If $M \redp N$, $P \redp Q$, and $x \not\in \bfv{M}$, then \[
    M[P/x] \redp N[Q/x]
  \]
\end{lem}

\begin{proof}
  By induction on the generation of $M \redp N$. The cases for
  most congruence rules and constants follow trivially, or from
  the IH. We prove the rest.
  \begin{indproof}
    \indcase{$\textsf{refl}$}
      Then $M \redp N$ is actually $M \redp M$, so we use Lemma
      \ref{lem:substredp} to infer $M[P/x] \redp M[Q/x]$. 

    \indcase{$\Box\textsf{int}$}
      Then $M \redp N$ is actually $\tilde f(\ibox{M}) \redp
      \ibox{f(M)}$. But $M$ and $f(M)$ are closed, so
      $\left(\tilde f(\ibox{M})\right)[P/x] \equiv \tilde
      f(\ibox{M}) \redp \ibox{f(M)} \equiv
      \left(\ibox{f(M)}\right)[Q/x]$.

    \indcase{$\supset_i$}
      Then $M \redp N$ is $\supset_G\ \textsf{true}\ M\ N \redp
      M'$ with $M \redp M'$. By the IH, $M[P/x] \redp M'[Q/x]$, so
      \[
	\supset_G\ \textsf{true}\ M[P/x]\ N[P/x]
	  \redp
	M'[Q/x]
      \] by a single use of $(\supset_1)$. The case for
      $\textsf{false}$ is similar.

    \indcase{$\rightarrow \beta$}
      Then $(\lambda x'{:}A.\ M)N \redp N'[M'/x']$, where $M \redp
      M'$ and $N \redp N'$. Then \[
	\left((\lambda x'{:}A.\ M)N\right)[P/x]
	  \equiv
	(\lambda x'{:}A.\ M[P/x])(N[P/x])
      \] But, by the IH, $M[P/x] \redp M'[Q/x]$ and $N[P/x] \redp
      N'[Q/x]$. So by $(\rightarrow \beta)$ we have \[
        (\lambda x'{:}A.\ M[P/x])(N[P/x])
          \redp
        M'[Q/x]\left[N'[Q/x]/x'\right]
      \] But this last is $\alpha$-equivalent to
      $\left(M'[N'/x']\right)\left[Q/x\right]$ by the substitution
      lemma.

    \indcase{$\Box\beta$}
      Then $\letbox{u'}{\ibox{M}}{N} \redp N'[M/u']$ where $N
      \redp N'$. By assumption, we have that $x \not\in \fv{M}$
      and $x \not\in \bfv{N}$. Hence, we have by the IH that
      $N[P/x] \redp N'[Q/x]$, so by applying $(\Box\beta)$ we get
      \begin{align*}
	(\letbox{u'}{\ibox{M}}{N})[P/x]
	  \equiv\
	    &\letbox{u'}{\ibox{M[P/x]}}{N[P/x]} \\
	  \equiv\
	    &\letbox{u'}{\ibox{M}}{N[P/x]} \\
	  \redp\
	    &N'[Q/x][M/u']
      \end{align*} But this last is $\alpha$-equivalent to
      $N'[M/u'][Q/x]$, by the substitution lemma and the fact that
      $x$ does not occur in $M$.

    \indcase{$\Box\textsf{fix}$}
      Then $\fixlob{z}{M} \redp{} M'[\ibox{(\fixlob{z}{M})}/z]$,
      with $M \redp{} M'$. As $x \not\in \bfv{\fixlob{z}{M}}$, we
      have that $x \not\in \fv{M}$, and by Lemma
      \ref{lem:varmon}, $x \not\in \fv{M'}$ either, so \[
        (\fixlob{z}{M})[P/x] \equiv \fixlob{z}{M}
      \] and\[
        M'[\fixlob{z}{M}/z][Q/x]
          \equiv
        M'[Q/x][\fixlob{z}{M[Q/x]}/z]
	  \equiv
	M'[\fixlob{z}{M}/z]
      \] Thus, a single use of $(\Box\textsf{fix})$ suffices.
  \end{indproof}
\end{proof}

We now pull the following definition out of the hat:
\begin{defn}[Complete development]
  The \emph{complete development} $M^\star$ of a term $M$ is
  defined by the following clauses:
    \begin{align*}
      x^\star
        &\defeq x \\
      c^\star
        &\defeq c \qquad (c \in 
	  \{\tilde f, \widehat{n},
	    \textsf{zero?}, \dots \}) \\
      \left(\lambda x{:}A.\ M\right)^\star
        &\defeq \lambda x{:}A.\ M^\star \\
      \left(\tilde{f}(\ibox{M})\right)^\star
	&\defeq \ibox{f(M)} \qquad \text{if $M$ is closed} \\
      \left(\left(\lambda x{:}A.\ M\right)N\right)^\star
        &\defeq M^\star[N^\star/x] \\
      \left(\supset_G\ \textsf{true}\ M\ N\right)^\star
        &\defeq M^\star \\
      \left(\supset_G\ \textsf{false}\ M\ N\right)^\star
        &\defeq N^\star \\
      \left(MN\right)^\star
        &\defeq M^\star N^\star \\
      \left(\ibox{M}\right)^\star
        &\defeq \ibox{M} \\
      \left(\letbox{u}{\ibox{M}}{N}\right)^\star
        &\defeq N^\star[M/u] \\
      \left(\letbox{u}{M}{N}\right)^\star
        &\defeq \letbox{u}{M^\star}{N^\star} \\
      \left(\fixlob{z}{M}\right)^\star
        &\defeq M^\star[\ibox{(\fixlob{z}{M})}/z]
    \end{align*}
\end{defn}

\noindent We need the following two technical results as well.

\begin{lem}
  \label{lem:reflstar}
  $M \redp{} M^\star$
\end{lem}
\begin{proof}
  By induction on the term $M$. Most cases follow immediately by
  (\textsf{refl}), or by the IH and an application of the relevant
  rule. The case for $\ibox{M}$ follows by $(\textsf{refl})$, the
  case for $\fixlob{z}{M}$ follows by $(\Box\textsf{fix})$, and
  the case for $\tilde f(\ibox{M})$ by $(\Box\textsf{int})$.
\end{proof}

\begin{lem}[BFV antimonotonicity]
  \label{lem:varmon}
  If $M \redp{} N$ then $\bfv{N} \subseteq \bfv{M}$.
\end{lem}
\begin{proof}
  By induction on $M \redp{} N$.
\end{proof}

\noindent And here is the main result:

\begin{thm}
  \label{thm:mainredp}
  If $M \redp P$, then $P \redp M^\star$.
\end{thm}

\begin{proof}
  By induction on the generation of $M \redp P$. The case of
  $(\textsf{refl})$ follows by Lemma \ref{lem:reflstar}, and the
  cases of congruence rules follow from the IH. We show the rest.
  \begin{indproof}
    \indcase{$\rightarrow\beta$}
      Then we have $(\lambda x{:}A.\ M)N \redp M'[N'/x]$, with $M
      \redp M'$ and $N \redp N'$. By the IH, $M' \redp M^\star$
      and $N' \redp N^\star$. We have that $x \not\in \bfv{M}$,
      so by Lemma \ref{lem:varmon} we get that $x \not\in
      \bfv{M'}$. Hence, by Lemma \ref{lem:redp} we get $M'[N'/x]
      \redp M^\star[N^\star/x] \equiv \left(\left(\lambda x{:}A.\
      M\right)N\right)^\star$.

    \indcase{$\Box\beta$}
      Then we have \[
	\letbox{u}{\ibox{M}}{N} \redp N'[M/u]
      \] where $N \redp N'$. By the IH, $N' \redp N^\star$, so it
      follows that \[
	N'[M/u] \redp N^\star[M/u] \equiv \left(
	\letbox{u}{\ibox{M}}{N}\right)^\star
      \] by Lemma \ref{lem:substint}.

    \indcase{$\Box\textsf{fix}$}
      Then we have \[
	\fixlob{z}{M}
	  \redp{}
	M'[\ibox{(\fixlob{z}{M})}/z]
      \] where $M \redp{} M'$. By the IH, $M' \redp{} M^\star$.
      Hence \[
	M'[\ibox{(\fixlob{z}{M})}/z]
	  \redp{}
	M^\star[\ibox{(\fixlob{z}{M})}/z]
	  \equiv
	\left(\fixlob{z}{M}\right)^\star
      \] by Lemma \ref{lem:substint}.
    \indcase{$\Box\textsf{int}$}
      Similar.
  \end{indproof}
\end{proof}

%% file: ipcfparallel.tex
\renewcommand{\arraystretch}{3}

\begin{center}
\begin{tabular}{c c}

  %% BETA

  $
    \begin{prooftree}
        \justifies
      M \redp{}M
        \using
      {(\textsf{refl})}
    \end{prooftree}
  $

  &

  $
    \begin{prooftree}
      M \redp{}N
        \quad\quad
      P \redp{}Q
        \justifies
      (\lambda x : A.\ M)P \redp{}N[Q/x]
        \using
      {(\rightarrow\beta)}
    \end{prooftree}
  $

  \\

  $
    \begin{prooftree}
      M \redp{}N
        \justifies
      \lambda x : A.\ M \redp{}\lambda x : A.\ N
        \using
      {(\textsf{cong}_\lambda)}
    \end{prooftree}
  $

  &

  $
    \begin{prooftree}
      M \redp{}N
        \quad\quad
      P \redp{}Q
        \justifies
      MP \redp{}NQ
        \using
      {(\textsf{app})}
    \end{prooftree}
  $

  \\

  $
    \begin{prooftree}
      P \redp{} P'
        \justifies
      \supset_G\ \textsf{true}\ P\ Q \redp{} P'
        \using
      {(\supset_1)}
    \end{prooftree}
  $

  & 

  $
    \begin{prooftree}
      Q \redp{} Q'
        \justifies
      \supset_G\ \textsf{false}\ P\ Q \redp{} Q'
        \using
      {(\supset_2)}
    \end{prooftree}
  $

  \\

  \multicolumn{2}{c}{
    $ \begin{prooftree}
      M \redp{} N
        \justifies
      \letbox{u}{\ibox{P}}{M} \redp{} N[P/u]
        \using
      {(\Box\beta)}
    \end{prooftree} $
  }

  \\

  \multicolumn{2}{c}{
    $
      \begin{prooftree}
	M \redp{} N
          \justifies
        \fixlob{z}{M} \redp{} N[\ibox{(\fixlob{z}{M})}/z]
          \using
        {(\Box\textsf{fix})}
      \end{prooftree}
    $
  }

  \\ 

  \multicolumn{2}{c}{
    $
      \begin{prooftree}
        \text{$M$ closed, $M \in \textrm{dom}(f)$}
          \justifies
        \tilde f(\ibox{M}) \redp{}\ibox{f(M)}	
          \using
        {(\Box\textsf{int})}
      \end{prooftree}
    $
  }

  \\

  \multicolumn{2}{c}{
    $
      \begin{prooftree}
        M \redp{}N
          \quad\quad
        P \redp{}Q
          \justifies
        \letbox{u}{M}{P} \redp{} \letbox{u}{N}{Q}
          \using
        {(\Box\textsf{let-cong})}
      \end{prooftree}
    $
  }

  \\

  \multicolumn{2}{c}{
    \begin{minipage}{\textwidth}
      \textbf{Remark}. In addition to the above, one should also
      include rules for the constants, but these are merely
      restatements of the rules in Figure \ref{fig:ipcfbeta}.
    \end{minipage}
  }

\end{tabular}
\end{center}